\documentclass[lettersize,journal]{IEEEtran}
\usepackage{settings}

\begin{document}

\title{Linear decomposition of approximate multi-controlled single qubit gates}
\author{
        Jefferson D. S. Silva\orcidlink{0000-0002-2063-6140}, 
        Thiago Melo D. Azevedo\orcidlink{0000-0002-1068-1618},
        Israel F. Araujo\orcidlink{0000-0002-0308-8701},
        Adenilton J. da Silva\orcidlink{0000-0003-0019-7694}%
\thanks{This work is based upon research supported by Conselho Nacional de Desenvolvimento Científico e Tecnológico - Brasil (CNPq) - Grant No. 409506/2022-2, No. 409513/2022-9 and No. 162052/2021-9. This study was financed in part by the Coordenação de Aperfeiçoamento de Pessoal de Nível Superior – Brazil (CAPES) – Finance Code 001, FACEPE (Grant No. APQ-1229-1.03/21), National Research Foundation of Korea (Grant No. 2022M3E4A1074591), and Yonsei University Research Fund of 2023 (2023-22-0072). T.M.D. Azevedo and J.D.S. Silva contributed equally to this work. Corresponding author: Adenilton J. da Silva}
\thanks{
T.M.D. Azevedo, J.D.S. Silva and A.J. da Silva are with Centro de Informática, Universidade Federal de Pernambuco, Recife, Pernambuco, 50.740-560, Brazil (e-mail: jdss2@cin.ufpe.br; tmda@cin.ufpe.br; ajsilva@cin.ufpe.br).}
\thanks{
I.F. Araujo is with the Department of Statistics and Data Science, Yonsei University, Seoul, Republic of Korea (e-mail: ifa@yonsei.ac.kr).}
}

\date{July 2023}

\maketitle

\begin{abstract}
    We provide a method for compiling approximate multi-controlled single qubit gates into quantum circuits without ancilla qubits. The total number of elementary gates to decompose an n-qubit multi-controlled gate is proportional to 32n, and the previous best approximate approach without auxiliary qubits requires 32nk elementary operations, where k is a function that depends on the error threshold. The proposed decomposition depends on an optimization technique that minimizes the CNOT gate count for multi-target and multi-controlled CNOT and SU(2) gates. Computational experiments show the reduction in the number of CNOT gates to apply multi-controlled U(2) gates. As multi-controlled single-qubit gates serve as fundamental components of quantum algorithms, the proposed decomposition offers a comprehensive solution that can significantly decrease the count of elementary operations employed in quantum computing applications.
\end{abstract}

\begin{IEEEkeywords}
Multi-Controlled Quantum Gates, Quantum Circuits, Quantum Compilation, Circuit Design, Approximated Quantum Gates
\end{IEEEkeywords}

\section{Introduction}
\IEEEPARstart{N}{oisy} Intermediate-Scale Quantum (NISQ)~\cite{preskill2018quantum} devices are currently in development. However, these devices are not universal quantum computers as they can only execute short-depth quantum circuits. A more efficient quantum device is crucial for practical quantum computing applications~\cite{kim2023evidence}. Quantum computation comprises both hardware and software components. One approach to enhance the efficiency of quantum devices is through improvements in the system software, enabling more efficient execution of quantum algorithms. Optimizing the system software of quantum devices is the primary focus of several studies, including error mitigation~\cite{kim2023evidence}, quantum state preparation~\cite{araujo2021divide, zhang2021low}, and the decomposition of unitary gates into quantum circuits~\cite{shende2006synthesis, iten2016quantum, malvetti2021quantum}.

The task of programming quantum computers involves breaking down quantum gates into sequences of elementary gates~\cite{barenco_1995,shende2006synthesis, iten2016quantum}. These elementary gates must encompass a universal set of quantum operations, such as single-qubit and CNOT gates~\cite{nielsen2010quantum}. Decomposing quantum gates into hardware-compatible elementary gates is a complex process that requires meticulous design and optimization~\cite{brugiere2021reducing, bae2020quantum, cuomo2023optimized, leymann2020bitter}. This research aims to minimize the number of elementary quantum operations essential for implementing an approximate multi-controlled quantum gate without auxiliary qubits.

The initial algorithm for decomposing multi-controlled $U(2)$ single-qubit gates without auxiliary qubits was proposed in Ref.~\cite{barenco_1995} and exhibited a growth in the number of elementary quantum operations and circuit depth that was quadratic about the number of controls.  Subsequently, Ref.~\cite{adenilton2022linear} proposed a decomposition with a quadratic number of elementary operations and linear depth.

While multi-controlled $SU(2)$ gates can be linearly decomposed into elementary gates, achieving a linear decomposition without auxiliary qubits for multi-controlled general $U(2)$ quantum gates remains an open problem. Numerous techniques make use of ancilla qubits to decrease the quantity of elementary quantum operations required for implementing a quantum gate~\cite{he2017decompositions, 10.1007/978-3-031-08760-8_16, iten2016quantum, barenco_1995}. Another approach is to introduce an approximation, up to an error $\epsilon$, for the multi-controlled $U(2)$ gates, which allows for decompositions with a linear number of basic operations and depth \cite{barenco_1995}.

This study introduces an alternative decomposition for constructing an approximate $n$-qubit controlled $U(2)$ gate using a linear number of elementary gates. Similarly to Ref.~\cite{barenco_1995}, the proposed method also includes an adjustable error parameter $\epsilon$, which can be optimized to either limit the circuit depth or adapt to the device noise level. Specifically, for this study, the error parameter is chosen to be upper-bounded by the order of magnitude of the error induced by a single CNOT gate in a real-world NISQ device. The decomposition proposed in this work performs better than the approximate decomposition from Ref.~\cite{barenco_1995}. 

Furthermore, this research introduces an optimization technique for implementing multi-controlled $X$ and multi-controlled $SU(2)$ gates with multiple targets (multi-controlled gates that share the same control qubits but have distinct target qubits). The optimization method effectively reduces the number of CNOT gates to apply multi-target gates. Consequently, this optimization reduces the CNOT gate count associated with the multi-controlled $U(2)$ gate approximation.

\begin{figure*}[t]
    \centering
    \includegraphics[width=.9\linewidth]{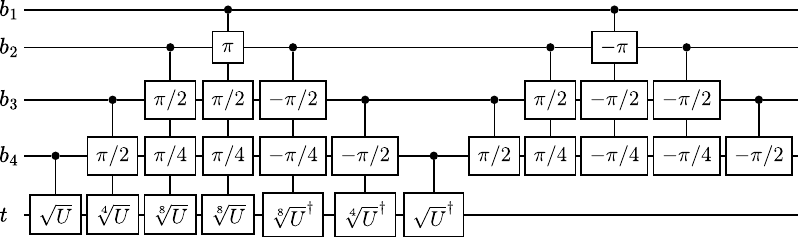}
    \caption{Decomposition of a multi-controlled $U(2)$ with four control qubits $\{ {b_1},  {b_2},  {b_3},  {b_4} \}$, and one target qubit $t$. Gates with angles $\pm \pi/2^k, k\in \mathbb{N}$, represent $R_x$ gates, and the angle is used as an argument, i.e., $R_x(\pm \pi/2^k)$.}
    \label{fig:adenilton_2022}
\end{figure*}

The rest of this work is structured as follows. Section~\ref{sec:relatedwork} provides an overview of related research in the field, focusing on multi-controlled gate decompositions without auxiliary qubits. Section~\ref{sec:proposedmethod} presents the main results: i) a linear approximation of multi-controlled quantum gates without auxiliary qubits and ii) an optimization approach for sequential multi-controlled gates featuring distinct target qubits, resulting in a reduction in the number of CNOT gates required by the approximated multi-controlled quantum gate decomposition. Section~\ref{sec:experiments} showcases experiments to validate the theoretical results, while Section~\ref{sec:conclusion} concludes the paper and hints at potential future research directions.

\section{Related Work}
\label{sec:relatedwork}

We start this section with a decomposition of multi-controlled $U(2)$ gates~\cite{adenilton2022linear}, followed by a decomposition of multi-controlled $SU(2)$ gates~\cite{adenilton_su2_2023} that will be used through this work. After that, we review an approximate decomposition of multi-controlled $U(2)$ gates from Ref.~\cite{barenco_1995} that will be compared to this work.

\subsection{Decomposition of multi-controlled U(2) gates}
\label{sec: original_linearu2}

Ref.~\cite{saeedi2013linear} proposed a method to decompose a multi-controlled $R_x$ gate with linear depth, and Ref.~\cite{adenilton2022linear} extended the method by replacing the $R_x$ gate, which operates on the target qubit, with an arbitrary operator $U\in U(2)$. This extension results in a method for decomposing $n$-qubit controlled $U$ gates with linear depth. However, this method still requires a quadratic number of CNOT gates.

The decomposition technique relies on the utilization of single-controlled $R_x$ gates and $2^j$-th roots of the $U$ gate, denoted as $U^{1/2^j}$, where $j\in \{1, 2, \dots, \nctrl-1 \}$ and $\nctrl$ represents the number of control qubits $\{ b_1, b_2, \dots, b_{n-1} \}$, where $\nctrl=n-1$ and the circuit features a single target qubit ($t=b_n$).

The multi-controlled circuit exhibits symmetry such that if any of the control qubits is 0, the applied gate is equivalent to the identity gate. Refer to \figref{adenilton_2022} for a visual representation of the circuit used for decomposing $C^3 U$, which involves three control qubits and one target qubit. The gates marked with angles of $\nicefrac{\pi}{j}$ represent $R_x$ gates with the respective angle as their argument (i.e., $R_x(\nicefrac{\pi}{j})$).

To precisely determine the CNOT count of the circuit, we can examine the number of controlled operators within the first and second triangles of the circuit. In the first triangle, the number of controlled gates is $\sum_{j=0}^{n-2} (1+2j)$, and in the second triangle, the summation becomes $\sum_{j=0}^{n-3} (1+2j)$. The sum of these values leads to $(n-1)^2+(n-2)^2$ controlled operators. Since each controlled $U(2)$ gate necessitates 2 CNOTs~\cite[Corollary 5.3]{barenco_1995}, the overall number of CNOT operations is $4n^2-12n+10$.

\subsection{Decomposition of multi-controlled SU(2) gates}

\begin{figure}[htb!]
    \centering
    \includegraphics[width=\linewidth]{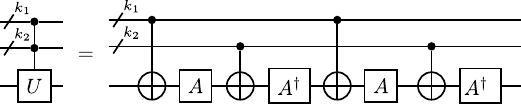}
    \caption{Decomposition scheme based  for multi-controlled $SU(2)$ gates, where $U, A \in SU(2)$, and $U$ has at least one real-valued diagonal.}  \label{fig:rafaella_circ1}
\end{figure}

For $SU(2)$ gates, various decompositions with a linear CNOT count have been proposed \cite{barenco_1995, iten2016quantum, adenilton_su2_2023}. Among these, Ref.~\cite{adenilton_su2_2023} is the most efficient approach that decomposes an $n$-qubit multi-controlled $SU(2)$ gate with a maximum of $20n-38$ ($20n-42$ for even $n$) CNOTs, where $n \geq 3$.

The decomposition of $SU(2)$ gates employs the circuit depicted in \figref{rafaella_circ1}. This allows the implementation of any unitary gate $U \in SU(2)$ with a real secondary diagonal with $16n-40$ CNOTs. By incorporating a Hadamard gate before and after the decomposed unitary gate, the circuit transforms into one capable of decomposing any gate with a real main diagonal. The same method decomposes multi-controlled SO(2) gates.

The decomposition of a general $SU(2)$ gate leverages the eigendecomposition of $U = QDQ^{-1}$, where $D$ is a diagonal matrix, and $Q$ is constructed from the eigenvectors of $U$. Given that $D$ is a diagonal matrix, it has at least one real diagonal. With eigenvectors phase selection, the matrix $Q$ will have a real main diagonal. Further optimizations enable the decomposition of a general $n$-qubit multi-controlled $SU(2)$ gate with a maximum of $20n-38$ ($20n-42$ for even $n$), where $n \geq 3$.

The $SU(2)$ decomposition presented here incorporates Lemma 7.2 from \cite{barenco_1995}, which outlines a decomposition technique for multi-controlled $X$ gates, $C^\nctrl X$. This method applies to a quantum circuit featuring $n \geq 5$ qubits and $\nctrl \in \{3, \dots, \lceil n/2 \rceil \}$, adopting a form illustrated in \figref{toffoli-chain-decomposition}, where the gates with the target $\oplus^{\bullet}$ represent the decomposition with relative phase~\cite{maslov2016advantages}, which is implemented using 3 CNOTs. In this work, we represent them as $C^2 X^{\bullet}$. 

\begin{figure}[htb!]
    \centering
    \includegraphics[width=\linewidth]{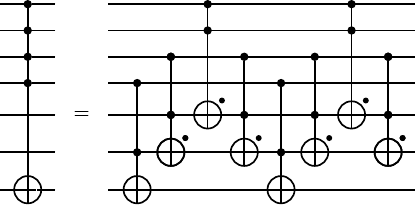}
    \caption{Decomposition of a $4$-controlled $X$ gate in a $7$-qubit system with two dirty auxiliary qubits \cite[Lemma 7.2]{barenco_1995}. In decomposition, the gates with targets $\oplus^{\bullet}$ represent the decomposition with relative phase~\cite[Fig 4]{adenilton_su2_2023}.}
    \label{fig:toffoli-chain-decomposition}
\end{figure}

\subsection{Approximate decomposition of multi-controlled U(2) gates}

\begin{figure}
    \centering
    \includegraphics[width=\linewidth]{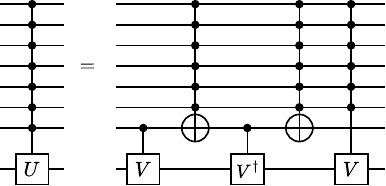}
    \caption{Lemma 7.5 from Ref.~\cite{barenco_1995}: Decomposition of a $(n-1)$-controllled $U(2)$ gate with $V^2=U$. It can be used recursively for the $(n-2)$-controlled $V$ gate, which leads to $O(n^2)$ elementary operations.}
    \label{fig:barenco_rec}
\end{figure}

Ref.~\cite[Lemma 7.8]{barenco_1995} proposed an approximate decomposition for $U(2)$ gates, up to an error $\epsilon$, by using the Lemma 7.5 of the same article, shown in \figref{barenco_rec}, recursively. The recursion stops after $k=\log_2(1/\epsilon)$ levels. Here, $\epsilon=\| C^{n-k-1}V_k - I\|$ and $V_j$ is the matrix used for the decomposition in the $j$-th level, where $V_j^2=V_{j-1}$. Ref.~\cite[Theorem 4]{iten2016quantum} makes the CNOT count analysis of the decomposition after applying some optimizations, where the CNOT cost of a $(n-1)$-controlled gate $U$ is given as a function of the CNOT cost of a $(n-2)$-controlled gate $V$ in addition to the two $C^{n-2}X$, which have a cost of $16n-40$ CNOTs each, and the two single-controlled gates that require $2$ CNOTs each. This leads to the following recurrence relation: 
\begin{equation}
    N_{\textnormal{CNOT}}(n-1) = 32n-76+N_{\textnormal{CNOT}}(n-2)
\end{equation}

Terminating this relation after the $k$-th level, i.e. when $C^{n-k-1}V_k\approx I$ and $N_{\textnormal{CNOT}}(n-k-1)=0$ leads us to the CNOT count of the decomposition of a $(n-1)$-controlled approximate gate:

\begin{equation}
     N_{\textnormal{CNOT}}(n-1) = -16k^2-60k+32nk
\end{equation}

Since $k$ depends solely on $\epsilon$, which is fixed, $k$ does not scale with $n$, and the decomposition has a linear number of CNOTs and a linear depth.

\section[Linear multi-controlled U(2) gates]{Linear multi-controlled $\mathrm{U(2)}$ gates}

\label{sec:proposedmethod}

This section presents a method for approximating the decomposition of an $\nctrl$-controlled quantum gate $C^\nctrl U$, where $U \in U(2)$. This approximation enables us to decompose multi-controlled quantum gates without auxiliary qubits using a linear number of elementary gates. To achieve this result, we leverage the limit $\lim_{\nctrl \to \infty} U^{1/2^{\nctrl-1}} = I$ and select a value for $\nctrl$ such that $U^{1/2^{\nctrl-1}}$ approximates $I$ up to an error $\epsilon$. 

\subsection[Approximating the N-th root of U]{Approximating the $N$-th root of $U$}

Let $U$ be a matrix in the unitary group $U(2)$ and $N$ a positive integer. The matrix $U^{1/N}$ approximates the identity matrix. Given an error $\epsilon$, we want to determine the value of $n_c = \log_2(N)$ such that $\|U^{1/{2^\nctrl-1}}-I\| \leq \epsilon$. 

\begin{lemma}\label{theorem: minimum_k}   For a given error $\epsilon$, $\|U^{1/2^{\nctrl-1}}-I\| \leq \epsilon$ if $\nctrl\geq \log_2 \left( \theta/{\arccos\left(1-\nicefrac{\epsilon^2}{2}\right)}\right)$.
 \end{lemma}

\begin{proof}
    Since $U$ is unitary, we can write $U = VDV^\dagger$, where $V$ is unitary and $D$ is diagonal. 
    Now consider $U^{1/N} = (VDV^\dagger)^{1/N}=VD^{1/N}V^\dagger$. 
    The diagonal elements of D are the eigenvalues ($e^{i\theta_1}$ and $e^{i\theta_2}$) of the matrix $U$. Eq.~\eqref{eq:n_root_D} describes the $N$-th root of $D$.
    \begin{equation}\label{eq:n_root_D}
        D^{1/N} =
        \begin{pmatrix}
            e^{i\theta_1/N} & 0 \\ 
            0 & e^{i\theta_2/N}
        \end{pmatrix}.
    \end{equation}
    
    As $N$ becomes large, $\theta_j/N$ approaches 0. Therefore, $e^{i\theta_1/N} \rightarrow 1$ and $e^{i\theta_2/N} \rightarrow 1$, which implies that $D^{1/N} \approx I$. Since $V$ is unitary, we know that $VV^\dagger = I$, we obtain the approximation in Eq.~\ref{eq:approx} when $N \gg 1$:
    \begin{equation}\label{eq:approx}
    U^{1/N} = (VDV^\dagger)^{1/N} \approx VV^\dagger = I.
    \end{equation}
    
    The error of this approximation is defined in Eq.~\ref{eq:error_approx}, where $\errorfunc$ is the error associated with the approximation (and $\epsilon$ is the error or tolerance we desire to obtain with the approximation). 
    
    \begin{equation}\label{eq:error_approx}
        \errorfunc \coloneqq \| U^{1/N}-I\|.
    \end{equation}
    
     For this work, the chosen metric $\| \cdot \|$ is the maximum absolute value of all the elements on a matrix, 
     \begin{equation}\label{eq:metric}
         \| U \|_{max} \coloneqq \max_{i,j} |U_{ij}|.
     \end{equation}
    
    In the diagonal basis
    \begin{equation}
        U^{1/N}-I = 
        \begin{pmatrix}
            e^{i\theta_1/N}-1 & 0 \\ 
            0 & e^{i\theta_2/N}-1
        \end{pmatrix},
    \end{equation}
    where $-\pi< \theta_1, \theta_2 \leq \pi$. The error $\errorfunc$ is given by the maximum absolute value of all the eigenvalues of the resulting matrix. In that way, it can be written as:
    \begin{equation}
    \errorfunc = | \lambda-1 |,
    \end{equation}
    where $\lambda$ is the eigenvalue of $U^{1/N}$ that maximizes $\errorfunc$. Thus, $\lambda$ is a complex value that can be expressed as $\lambda=\lambda_r + i\lambda_i$ and since $U^{1/N}$ is unitary, its eigenvalues must have an absolute value equal to 1. Therefore, $\lambda_i^2 = 1 - \lambda_r^2$, which results in:
    
    \begin{equation}\label{eq:erro2}
        \errorfunc^2 =  (\lambda_r -1)^2 + 1 - \lambda_r^2 = 2(1-\lambda_r).
    \end{equation}
     Since $\lambda= e^{i \theta/N}$, $-\pi < \theta \leq \pi$, then $\lambda_r= \cos(\theta/N)$ and $\errorfunc^2 = 2(1-\cos(\theta/N))$.
    We can then write the error as:
    \begin{equation}
        \errorfunc = \sqrt{2(1-\cos(\theta/N))}.
    \end{equation}
    Therefore, $N$ can be written as a function of $\errorfunc$:
    \begin{equation}\label{eq:base_error}
        N =  \frac{|\theta|}{ \arccos(1-\frac{\errorfunc^2}{2})}.
    \end{equation}
    As $N = 2^{\nctrl-1}$, if $\nctrl\geq \log_2 \left( |\theta|/{\arccos\left(1-\nicefrac{\epsilon^2}{2}\right)}\right)+1$, then the error of the approximation is $\errorfunc\leq\epsilon$. 
\end{proof}

The approximation error can be set arbitrarily small by increasing $\nctrl$, while $\theta$ depends on the original $U$ operator. Given a fixed $\epsilon$, the worst case scenario is when $\theta = \pi$, which requires a higher value for $\nctrl$. For example, the IBM Mumbai quantum device~\cite{ibm_mumbai} has a median CNOT error of $8.255e-3$. Therefore, $\nctrl$ can be chosen to be sufficiently large so that the error of the approximation is smaller than the median CNOT error of actual quantum devices. For $\theta=\pi$, $U=X$, we can achieve a tolerance equal to or lower than $1e-3$ with $\nctrl=13$.

\subsection{Linear decomposition of multi-controlled gates}
Using the decomposition from Section~\ref{sec: original_linearu2} and the approximation of Lemma~\ref{theorem: minimum_k}, we can construct a decomposition of a $U(2)$ gate with a linear number of CNOTs. In the first step, we truncate the decomposition of the controlled gate using a fixed number of base control qubits $\nbase < \nctrl$ and build an ($\nbase+1)$-qubit multi-controlled $U(2)$ gate. In the second step, we add extra qubits on top of the truncated decomposition to obtain the $\nctrl$-controlled $U(2)$ gate, and in the third step, we make use of the approximation to reduce the CNOT count of the circuit. 

We build the $\nbase$-controlled $U(2)$ gate with the method proposed in Ref.~\cite{adenilton2022linear}. Since $\nbase$ is fixed, to increase the number of controls to $\nctrl$, we add $\nextra$ control qubits $\{{e_1}, {e_2}, \ldots, {e_\nextra}\}$ on top of $\nbase$ existing qubits $\{ {b_1}, {b_2}, \ldots, {b_\nbase}\}$, where $\nctrl=\nbase+\nextra$. These additional control qubits act solely on the gates with $b_1$  as a control qubit in the $\nbase$-controlled $U(2)$. \figref{add_controls} shows an example of this circuit for $\nbase=4$ and $\nextra=2$. 

\begin{figure*}[t]
    \centering
    \includegraphics[width=.9\linewidth]{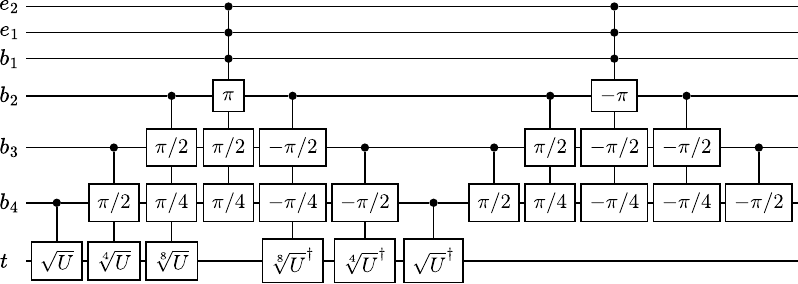}
    \caption{Approximate decomposition of a multi-controlled $U(2)$ gate with one target qubit $t$ and six control qubits. These control qubits include two extra control qubits $\{ e_1$, $e_2 \}$ and four base control qubits $\{b_1$, $b_2$, $b_3$, $b_4 \}$. The gates marked with angles $\pm \pi/2^k$, where $k\in \mathbb{N}$, represent $R_x$ gates, with the specified angle used as an argument, i.e., $R_x(\pm \pi/2^k)$.}
    \label{fig:add_controls}
\end{figure*}

\begin{lemma}\label{theorem: add_qubit}
    It is possible to create a $C^\nctrl U$ gate, $U \in U(2)$ by using the decomposition of an $(\nbase+1)$-qubit multi-controlled gate $C^{\nbase} U$ from Ref.~\cite{adenilton2022linear} and adding $\nextra$ control qubits, $\nctrl=\nbase+\nextra$, that act as additional control qubits for the gates that have $b_1$ as a control qubit.
\end{lemma}

\begin{proof}
    We can take a $C^\nbase U$ with $\nbase+1$ qubits, $\{{b_1}, {b_2}, \ldots, {b_\nbase}, t\}$ and add $\nextra$ extra control qubits $\{ {e_1}, {e_2}, \dots, {e_\nextra} \}$. These $\nextra$ control qubits act as controls only for the gate with ${b_1}$ as a control qubit. If any of the $n_e$ control qubits is $0$, then it will give the same result as if ${b_1}$ is $0$, and since the circuit was originally a $C^{\nbase} U$ gate, the target qubit is not altered. If all $n_e$ control qubits are equal to $1$, then the action on the target qubit still depends on the $n_b$ control qubits. If any of the $n_b$ qubits is $0$, the target qubit is unaltered~\cite[Theorem 1]{adenilton2022linear}. But, if they all are equal to $1$, then the $U$ gate is applied~\cite[Theorem 1]{adenilton2022linear}. Therefore, the circuit constructed is a $C^\nctrl U$ gate. 
\end{proof}

The result of Lemma~\ref{theorem: add_qubit} allows us to decompose $C^\nctrl U$, but the CNOT cost of the decomposition is still quadratic in the number of control qubits. To address this, we can use the results of Lemma~\ref{theorem: minimum_k} to replace the central multi-controlled $U^{1/2^{\nbase-1}}$ gate that has the additional $\nextra$ control qubits with an identity matrix, with an associated error $\epsilon$, given that we have a sufficiently large $\nbase$. After discarding the central multi-controlled $U^{1/2^{\nbase-1}}$ gate, the remaining multi-controlled gates on the circuit are $R_x$ gates that are $SU(2)$ gates and can be decomposed with an asymptotically linear number of CNOTs~\cite{adenilton_su2_2023}). We have a fixed number of single-controlled $U^{1/2^j}$ and $R_x$ gates that depend on $\nbase$, which is fixed and depends only on the error $\epsilon$. Therefore, if we use this approximation, the number of CNOTs will increase linearly with $\nctrl$. We can also calculate the exact number of CNOTs used by the decomposition.

\begin{theorem}\label{theorem: new_cnot_count}
    An approximated decomposition of a $C^\nctrl U$ gate, $U\in U(2)$, can be constructed, up to an error $\epsilon$ with an upper bound of $-28(\nbase-1)^2+2(\nbase-1)(16n-40)$ CNOTs, $n\geq\ \nbase$, in which $\nbase$ depends on $\epsilon$.
\end{theorem}

\begin{proof}

We can again divide the circuit into two triangles and count the number of controlled operators in each triangle. For the first triangle, ignoring the central column of multi-controlled gates, in each line, the number of single-controlled operators can be given as $\{2j \in \mathbb{Z} : 0 \le j \le \nbase-1\}$, summing all the lines gives us $\sum_{j=1}^{\nbase-1} 2j = (\nbase)(\nbase-1)$. For the second triangle, we have similarly $\sum_{j=1}^{\nbase-2} 2j = (\nbase-1)(\nbase-2)$, summing both we have $2(\nbase-1)^2$. As each operator requires 2 CNOTs \cite{barenco_1995}, we have a CNOT count of $4(\nbase-1)^2$ from these gates.

Now, looking at the multi-controlled gates of the central column in the first triangle, the controlled $\nrootU$ is removed, and what remains are $\nbase-1$ multi-controlled $R_x$ gates with $\nextra+1$ control qubits each. These multi-controlled gates can be implemented with at most $16(\nextra+2)-40$ CNOTs each \cite{adenilton_su2_2023}, which leads to a cost of $2(\nbase-1)[16(\nextra+2)-40]$ considering both triangles. We can rewrite it in terms of $n$ instead of $\nextra$, and since $n=\nctrl+1=\nbase+\nextra+1$, this will give us $2(\nbase-1)(16n-40)-32(\nbase-1)^2$. Then, the total CNOT count of the circuit is 
\begin{equation}
   N_{\textnormal{CNOT}} = -28(\nbase-1)^2 + 2(\nbase-1)(16n-40).
\end{equation}

\end{proof}

\noindent Considering that $\nbase$ is fixed and depends solely on the error $\epsilon$, we can write the corollary:

\begin{corollary}
    An approximated decomposition of a $C^\nctrl U$ gate, $U\in U(2)$, can be constructed, up to an error $\epsilon$ with a CNOT count that is linear in n, by using the decomposition proposed in Lemma~\ref{theorem: add_qubit}
\end{corollary}

While the outcome is linear with respect to $n$, the introduction of multi-controlled $SU(2)$ gates brings with it high coefficients. This means the approximated circuit ends up using more CNOTs than its original counterpart for smaller values of $n$. To enhance the efficiency of the proposed method, it is crucial to minimize the CNOT count of the multi-controlled $R_x$ gates, as discussed in Section~\ref{sec:multitarget}. Nonetheless, Theorem 1 outperforms the approximate decomposition proposed in Ref.~\cite{barenco_1995}. When leveraging the optimizations from Ref.~\cite{iten2016quantum}, as depicted in \figref{barenco_compare}, Theorem 1 maintains a consistent advantage in CNOT count, especially evident when $\epsilon=0.001$.

\begin{figure}
    \centering
    \includegraphics[width=\linewidth]{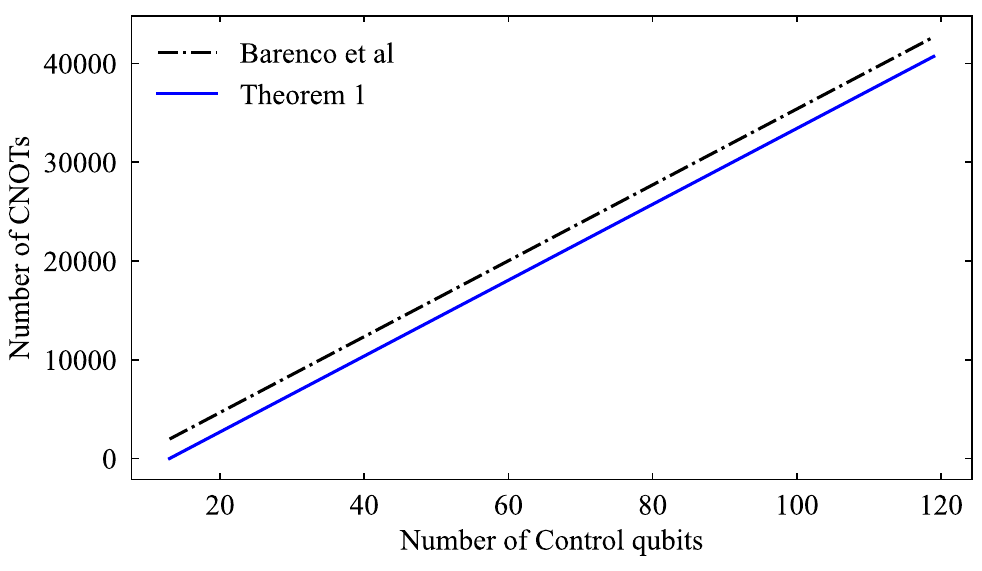}
    \caption{Comparison on the number of CNOTs between the method proposed Theorem 1 from this work against Lemma 7.8 from Ref.~\cite{barenco_1995}, for $\epsilon=0.001$.}
    \label{fig:barenco_compare}
\end{figure}

\subsection{Error analysis of the approximated circuit}

While the error of approximating a single gate is $\errorfunc$, we would like to determine the total error of the entire circuit when we use the approximation. For this, we need to consider two cases: In the first case, all the control qubits are active, and the original circuit (without approximation) should act with an $U$ gate in the target qubit; In the second case, at least one of the control qubits is inactive, and the original circuit should act with an identity matrix $I$ on the target qubit. In both cases, the error is defined using the same metric from Eq.~\eqref{eq:metric}. We have the Lemma~\ref{theorem: error_circuit_active} for the first case.

\begin{lemma}{\label{theorem: error_circuit_active}}
When all controls are active, the approximated circuit will have a total error $\errorfunc_T = \errorfunc$ concerning the original circuit, which should apply an $U$ gate on the target qubit.
\end{lemma}
\begin{proof}

By using the approximation, the output of the circuit when all controls are active will be  $UU^{-1/2^{\nbase-1}}$. Then, the total error of the circuit $\errorfunc_T$ is given by:

\begin{equation}
    \errorfunc_T = \| U U^{-1/2^{\nbase-1}}-U\|_{max}.
\end{equation}

\noindent In the diagonal basis 
\begin{equation*}
    \| U(U^{-1/2^{\nbase-1}}-I) \|_{max} = \|U\|_{max} \cdot \| \nrootInv-I \|_{max}.
\end{equation*}
Since $U \in U(2)$ and $\|U\|=1$: 
\begin{align}
    \errorfunc_T 
    &= \|\nrootInv-I\|_{max} \nonumber \\
    &=\|\nrootInv(I-\nrootU)\|_{max} \nonumber \\
    &= \|\nrootU-I\|_{max}=\errorfunc.
\end{align}
So, when all the controls are active, the total error of the circuit is equal to the individual error of the approximated gate.

\end{proof}

\begin{lemma}{\label{theorem: error_circuit_inactive}}
When at least one of the controls is inactive, the approximated circuit will have a total error $\errorfunc_T = \errorfunc$ concerning the original circuit, which should apply an $U$ gate on the target qubit.
\end{lemma}

\begin{proof}

If any of the control qubits $b_1$ or $e_1, e_2, \dots, e_\nextra$ are inactive, the circuit will perform the same operation as it would without the approximation, as it only affects the gate controlled by these qubits, and, as per~\cite{adenilton2022linear}, the output will be the identity, and the error $\errorfunc=0$.

If the control qubits $b_1$ and $e_1, e_2, \dots, e_\nextra$ are active, and one of the other control qubits is inactive, then there will be a $\nrootU$ missing for the circuit to output the identity. Therefore, the output would be $U^{-1/2^\nbase}$. The total error $\errorfunc_T$ is then given by Eq.~\eqref{eq:t_error}. Which is the same error as Lemma~\ref{theorem: error_circuit_active}.
\begin{equation}\label{eq:t_error}
    \errorfunc_T = \| \nrootInv-I\|_{max}.
\end{equation}

\end{proof}

\subsection{Multi-controlled multi-target SU(2) gates decomposition}
\label{sec:multitarget}

According to Ref. \cite{adenilton_su2_2023}, we can provide the decomposition of any $n$-qubit multi-controlled $W \in SU(2)$ operator with a number of gates proportional to 20n (or 16n if $W$ has a real diagonal). This decomposition is precisely expressed in terms of $A \in SU(2)$ operators and $C^{k_1} X$ and $C^{k_2} X$ gates, where $k_1 = \lceil k/2 \rceil$ and $k_2 = \lfloor k/2 \rfloor$.
In this section, we extend these theorems to encompass the decomposition of a collection of $\ntargets$ multi-controlled one-qubit operators $W_{t_i} \in SU(2)$. To start, we introduce the notation $\prod_{i=1}^{\ntargets} C^k M_{t_i}$ to describe a collection of $\ntargets$ one-qubit operators $M_{t_i}$ that share the same control qubits while having distinct target qubits.

First, we present a multi-target version of the Toffoli gate. Ref.~\cite{nielsen2010quantum} showed us the Toffoli gate decomposition \figref{originalToffoli} using six CNOTs. However, we can create a multi-target version of this gate, where each new target leads to a fixed cost of 4 CNOTs.

\begin{figure}[htb!]
    \centering
    \includegraphics[width=\linewidth]{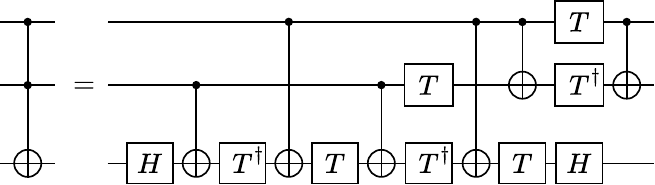}
    \caption{Exact decomposition of the Toffoli gate, described by \cite{welch2014efficient}, where the target $ X = H \cdot Z \cdot H$.}
    \label{fig:originalToffoli}
\end{figure}

\begin{lemma}\label{multi-target-toffoli}
    The gates $\prod_{i=1}^{\ntargets} C^2 X_{\target_i}$ can be implemented using $2\ntargets + 4$ CNOTs.
\end{lemma}

\begin{proof}
    Consider the $C^2 X$ decomposition described in \cite{welch2014efficient}, as shown in \figref{originalToffoli}, where $T = \begin{pmatrix}1 & 0\\ 0 & e^{i \pi/4}\end{pmatrix}$ is the matrix that induces a phase $\pi/4$. This decomposition attributes the X gate action to the target qubit.
    
    If we want to add a new target for this decomposition, we can introduce two $CX$ gates on the new target, with control qubits on either side of an $X$ gate. This way, the second target qubit activates only when the first target does. In this scenario, only one control qubit will trigger the $X$ gate on the second target.
    
    It is a recurring result: we add two $CX$ gates to the following target with control qubits on both sides of the preceding $CX$ gates. The initial target qubit requires six CX gates. Thus, the total CNOT count for a Toffoli multi-target $\prod_{i=1}^{\ntargets} C^2 X_{\target_i}$ is
    
    \begin{equation}
        N_{\textnormal{CNOT}} = 2\ntargets + 4.
    \end{equation}
\end{proof}

Lemma \ref{multi-target-toffoli} is a reformulation of the multi-target version of the Toffoli decomposition described in Ref.~\cite{wille2013improving}. Reference~\cite{wille2013improving} uses a decomposition for any operator $U$ with two controls, requiring 8 CNOTs. However, we can use the result of Ref.~\cite{welch2014efficient} to decompose a Toffoli gate using only $H$ and $T$-gates, requiring only 6 CNOTs.

We will use Lemma \ref{multi-target-toffoli} to construct a decomposition for the gates $\prod_{i=1}^{\ntargets} C^{\nctrl} X_{\target_i}$, extending Lemma 7.2 in \cite{barenco_1995} to describe their decomposition. \figref{multi-target_mcx} ilustrates Lemma~\ref{lem:mcx-multi-target} with $4$ controls, $3$ targets, and $2$ auxiliary qubits.

\begin{lemma}\label{lem:mcx-multi-target}
   The gates $\prod_{i=1}^{\ntargets} C^{\nctrl} X_{\target_i}$ can be generated by a circuit with $2(2\nctrl + \ntargets - 5)$ $C^2 X$ gates, with the restriction that the total number of qubits $n \geq 5$ and $\nctrl \in \{3, \dots, \lceil n/2 \rceil \}$.
\end{lemma}

\begin{proof}
    
    The circuit is very similar to the one of Lemma 7.2 from Ref.~\cite{barenco_1995}, with the gates that act only on the control and auxiliary qubits being unaltered. The modification is in the addition of $(\ntargets-1)$ $C^2 X$ gates that have the same control qubits as the $C^2 X$ targeting the $t_1$ target qubit, but act on the target qubits $\{ t_2, t_3, \dots, t_{\ntargets} \}$. And, they are all placed together with the $C^2 X$ gate targeting the $t_1$ qubit, all of them need to be applied before the other gates that act exclusively on the control or auxiliary qubits. 
    
    Since the gates that act only on the control and auxiliary qubits are unaltered if one of the control qubits is inactive, then the action of the gates targeting $\{ t_2, t_3, \dots, t_{\ntargets} \}$ will have the same result as the one targeting $t_1$. If one of the controls is not active, then the target qubits will be unchanged. If all the control qubits are active, then the action of the gates targeting $\{ t_2, t_3, \dots, t_{\ntargets} \}$ will again have the same result as the one targeting $t_1$, in all the target qubits, we will have the action of a $X$ gate.
\end{proof}

\begin{figure}[htb!]
    \centering
    \includegraphics[width=\linewidth]{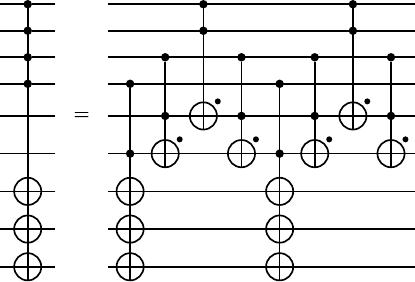}
    \caption{Decomposition of a multi-controlled multi-target CNOT gate with $4$ controls, $3$ targets, and $2$ auxiliary qubits: The gates that act on the additional target qubits $\{ t_2, t_3 \}$ are placed together with the gate that targets the $t_1$ qubit and all of them need to be applied before the other gates that act exclusively on the control or auxiliary qubits. The gates that act only on the control and auxiliary qubits are unaltered.}
    \label{fig:multi-target_mcx}
\end{figure}

\begin{figure*}[ht!]
    \centering
    \includegraphics[width=.7 \linewidth]{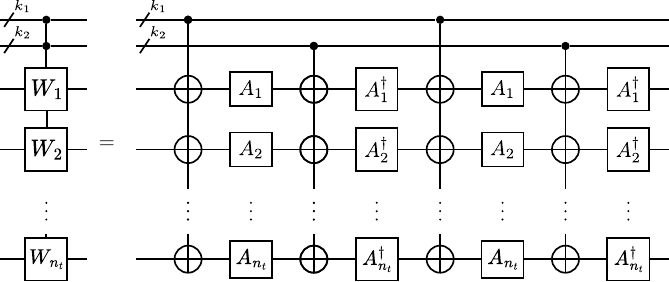}
    \caption{Decomposition scheme based for multi-controlled multi-target $SU(2)$ gates, where $W_{i}$, with $i = \{1, 2,\dots,\ntargets\}$ and $ W \in SU(2)$, has at least one real-valued diagonal.}
    \label{fig:setsu2}
\end{figure*}

The decomposition of Lemma \ref{lem:mcx-multi-target} works in two stages: the action stage, which promotes the action of control qubits on each target qubit, and the reset stage, to return the auxiliary qubits to the original state. The number of gates $C^2 X^\bullet$, acting on the auxiliary qubits, scale linearly with the number of control qubits.

For the lowest control qubit level, $n_c = 3$, we need two gates $C^2 X^\bullet$, one for the action stage, and another to reset the action stage. For each new control, we need to add four gates $C^2 X^\bullet$, two for the action stage and two for the reset stage.

However, because we do not have enough Toffoli gates to cancel the matrix of the relative phase, we can not implement the other Toffoli gates (those that act on the target qubits, $\prod_{i=1}^{\ntargets} C^k X_{\target_i}$) using the relative phase decomposition. Instead, we need to apply the gates $\prod_{i=1}^{\ntargets} C^2 X_{t_i}$ with CNOT cost described on Lemma \ref{multi-target-toffoli}.

Lemma~\ref{lem:mcx-multi-target} can then be used together with the decomposition of $SU(2)$ gates from Ref.~\cite{adenilton_su2_2023}, to modify the decomposition for a multi-target scenario.

\begin{lemma}\label{exact_SU(2)_multi-target}
    $\prod_{i=1}^{\ntargets} C^k W_{\target_i}$ can be generated by the circuit of the \figref{setsu2}, if the operators $W_{\target_i}$ have real elements on one of the diagonals, and is built from four elements: a set of gates $A \in SU(2)$~\cite[Lemma 1]{adenilton_su2_2023}, and its conjugate and transpose gates, $C^{k_1} X$ and $C^{k_2} X$, where $k_1 = \lceil k/2 \rceil$ and $k_2 = \lfloor k/2 \rfloor$, described in \lemref{mcx-multi-target}.
\end{lemma}

\begin{proof}
    In the circuit of \figref{setsu2}, the $A_i$ and $A_i^{\dagger}$ gates affect only the $\target_i$ qubits. Additionally, according to \lemref{mcx-multi-target}, the $C^{k_1} X_{\target_i}$ and $C^{k_2} X_{\target_i}$ targets also exclusively affect the $\target_i$ qubits. Consequently, each $\target_i$ qubit undergoes the action of the $U$ gate decomposition associated solely and exclusively with itself~\cite[Lemma 1]{adenilton_su2_2023}.
\end{proof}

Lemma 6 will be used to optimize the multi-controlled $U(2)$ gate, implementing all multi-controlled $R_x$ gates (see \figref{add_controls}). In this case, since the $R_x$ gates have real values on the main diagonal, we will apply~\cite[Theorem 2]{adenilton_su2_2023}, making a modification to each $R_x$ gate, allowing them to be decomposed in the form of Lemma 6. The following theorem describes the circuit complexity of Lemma \ref{exact_SU(2)_multi-target}.

\begin{theorem}\label{theorem: multi-target_su2}
    With $n\geq 8$, the decomposition of a multi-target $(n-2)$-controlled $SU(2)$ gate with $\ntargets$ targets requires $16n+16(\ntargets-1)-40$ CNOTs.
\end{theorem}

\begin{proof}
    Consider the gates $\prod_{i=1}^{\ntargets} C^k W_{\target_i}$. For the first target $C^k W_{\target_1}$, we can use~\cite[Theorem 3]{adenilton_su2_2023} to account for $16n - 40$ CNOTs. For the subsequent $\prod_{i=2}^{\ntargets} C^k W_{\target_i}$, each $C^k W_{\target_i}$ requires 2 $\prod_{i=2}^{\ntargets} C^{k_1} X_{\target_i}$ and 2 $\prod_{i=2}^{\ntargets} C^{k_2} X_{\target_i}$, each implemented using $4(\ntargets-1)$ CNOT gates. Therefore, the total number of CNOTs is $16n + 16(n_t - 1)- 40$.
\end{proof}

\begin{figure}[htb!]
    \centering
    \includegraphics[width=\linewidth]{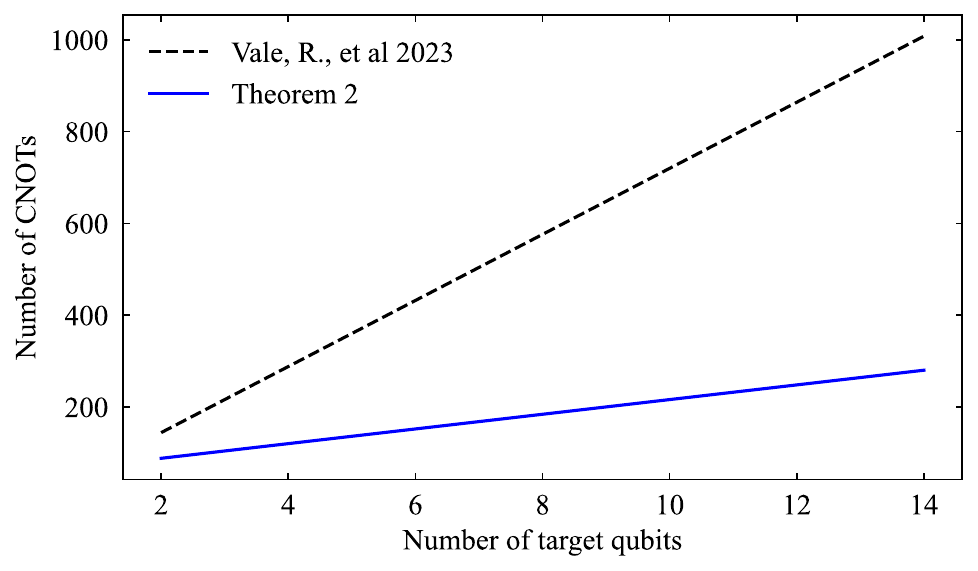}
    \caption{Analytical comparison of the number of CNOTs in the multi-controlled gate using multiple copies of the original $SU(2)$ decomposition scheme \cite{adenilton_su2_2023}, represented by the black dashed line, and the proposed multi-target version method, represented by the blue solid line. The graph displays the CNOT count for different numbers of target qubits, all with six control qubits.}
    \label{fig:su2_cascade_targets}
\end{figure}

For the original $SU(2)$ decomposition, the CNOT count for multiple targets (one multi-controlled gate for each target) would be $\ntargets(16n-40)$. Compared to the result of Theorem~\ref{theorem: multi-target_su2}, for $6$ controls or more, it is advantageous to use the multi-target optimization from Theorem~\ref{theorem: multi-target_su2} instead of applying a copy of gate on each different target. 

\figref{su2_cascade_targets} shows the comparisons of the number of CNOTs according to the number of targets of the original method~\cite{adenilton_su2_2023} (making multiple copies of the multi-controlled gate) versus the proposed multi-target method. The graphic shows the CNOT count for different numbers of target qubits, all with six control qubits.

\subsection{Optimization for general k-controlled U(2) gates}

We can further optimize our approximation for a $k$-controlled $U(2)$ gate using the results of Theorem~\ref{theorem: multi-target_su2} to reduce the number of CNOTs added by the multi-controlled $R_x$ gates.

\begin{theorem}\label{last_theorem}
    An approximated decomposition of a $C^\nctrl U$ gate, $U\in U(2)$, can be constructed up to an error $\epsilon$ with a CNOT count upper bound of $4(\nbase-1)^2 + 32n - 112$, in which $\nbase$ is a fixed base number of qubits that depends on $\epsilon$, as in Eq.~\eqref{eq:base_error} and $n\geq\nbase+8$.
\end{theorem}

\begin{proof}
    The proof is similar to the one of Theorem~\ref{theorem: new_cnot_count}, and the difference is in the CNOT count of the multi-controlled $R_x$ gates. We have $\nbase-2$ multi-controlled $R_x \in SU(2)$ gates with the same $(\nextra+1)$ control qubits for each of the two triangles. For $n\geq\nbase+8$ we can use Theorem~\ref{theorem: multi-target_su2}, and these $\nbase-2$ multi-controlled gates can be tought as a multi-target $(\nextra+1)$-controlled $SU(2)$ gate, so they add $2[16(\nextra+2)+16(\nbase-2)-40]$ CNOTs after summing both triangles, or $32\nextra+32\nbase-80$ CNOTs after refactoring. Adding those to the $4(\nbase-1)^2$ from the multi-controlled $U^{1/2^j}$ gates and writing in terms of $n$ and $\nbase$, where $n=\nbase+\nextra+1$, the total cost of the circuit is
    \begin{equation}
        N_{\textnormal{CNOT}} = 4(\nbase-1)^2 + 32n - 112.
    \end{equation}
\end{proof}
 
With this result, we can reduce the multiplicative constant of our CNOT upper bound and make our decomposition more efficient. In \figref{u2_optim_comparison}, we compare the upper bound of CNOTs for multi-controlled $U(2)$ gates of Theorems \ref{theorem: new_cnot_count} and \ref{last_theorem} versus the original exact decomposition from Ref.~\cite{adenilton2022linear}. In this figure, it is clear that as the number of control qubits grows, Theorems \ref{theorem: new_cnot_count} and \ref{last_theorem}, which have linear CNOT costs, perform significantly better than the quadratic exact decomposition. 
We can observe that Lemma \ref{lem:mcx-multi-target} significantly reduces the CNOT cost of the approximated multi-controlled $U(2)$ gate, and gives Theorem \ref{last_theorem} a significant advantage over Theorem \ref{theorem: new_cnot_count}, which already was more efficient than the approximated decomposition from Ref.~\cite{barenco_1995}.

\begin{figure}[t]
    \centering
    \includegraphics[width=\linewidth]{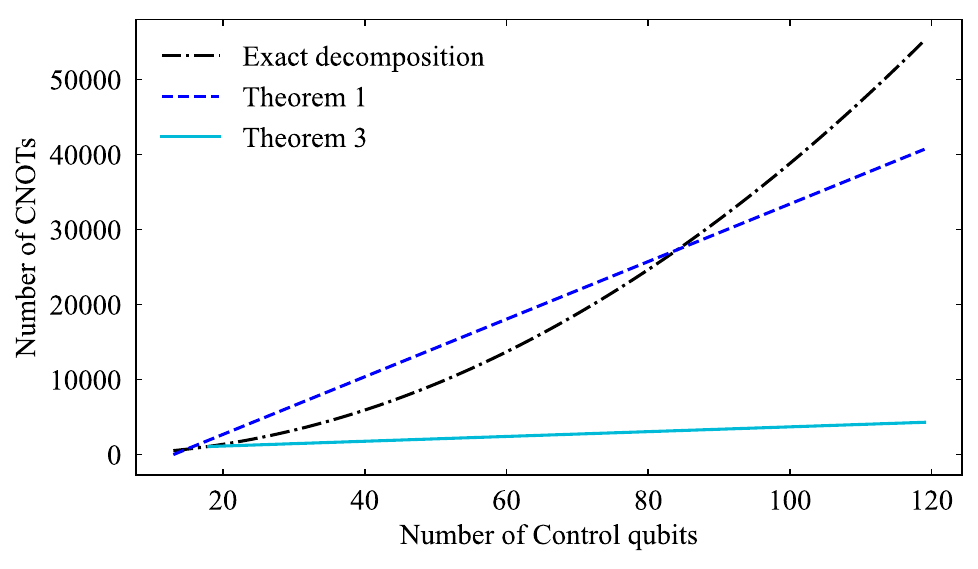}
    \caption{Comparison of the mathematical upper bound for the count of CNOTs of multi-controlled $U(2)$ gates. The exact linear $U(2)$ gate decomposition \cite{adenilton2022linear} is represented by the black dash-dot curve, the approximated decomposition without multi-target optimization is shown as the dark blue dash curve, and the light blue solid line depicts the approximated decomposition with optimization. For the approximated versions, $\nbase=13$.}
    \label{fig:u2_optim_comparison}
\end{figure}

\section{Experiments}

\label{sec:experiments}

In this section, we experimentally compare the CNOT count of our decomposition with the original multi-controlled linear-depth $U(2)$ gate decomposition~\cite{adenilton2022linear}. We consider the approximation error $\epsilon=10^{-3}$ in all experiments because this value is smaller than the current median CNOT error of IBM Mumbai Quantum Computer~\cite{ibm_mumbai}.
We experiment with a multi-controlled $X$ gate because it has an eigenvalue -1, which is the worst case for the approximated multi-controlled gate decomposition. We utilize the qiskit (version 0.43.2) and qclib (version 0.0.23) Python libraries for the experiments. After applying the proposed method, we decompose the multi-controlled gate into single qubit and CNOT gates without optimization.

\figref{cnot_count_mcx} compares the theoretical results for both the exact and the approximate decomposition to an experimental result, analyzing the decomposition CNOT cost for multi-controlled gates $C^n X$. In a numerical experiment, the proposed method and Theorem~\ref{last_theorem} assume the same value starting from $\nbase+5$ controls. Prior to entering the linear regime, where the number of controls for the multi-target multi-control $SU(2)$ gates is less than 5, we apply the multi-target Toffoli gate from Lemma~\ref{multi-target-toffoli} instead of the Lemma~\ref{lem:mcx-multi-target} for general multi-controlled $X$ gates, which results in a reduced CNOT count. 
Consequently, when compared with the exact decomposition, this approach allows us to achieve a more cost-effective CNOT operation.

\begin{figure}[ht]
     \centering
     \includegraphics[width=\linewidth]{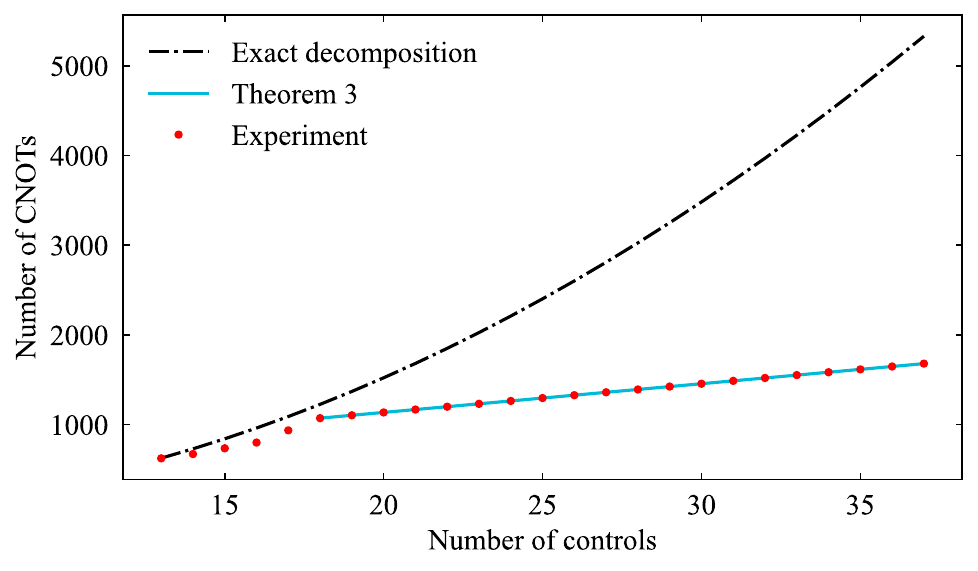}
     \caption{Comparison of Quantum Circuit Decomposition CNOT cost for multi-controlled $U=X$ gates using different methods. The light blue solid line represents the analytical cost of CNOTs as shown in Theorem~\ref{last_theorem}, the red dotted curve represents our experimental CNOT cost, and the black dashed-dotted line represents the exact decomposition~\cite{adenilton2022linear}.}
     \label{fig:cnot_count_mcx}
\end{figure}

\section{Conclusion}
\label{sec:conclusion}

In this paper, we presented a decomposition technique for approximate multi-controlled $U(2)$ gates, demonstrating a CNOT count that scales linearly with the number of qubits. The proposed methodology involves approximating the $2^{\nctrl-1}$-root of the initial $U(2)$ gate relative to the identity matrix while maintaining an approximation error $\epsilon$ that can be arbitrarily small. 

Despite achieving linear growth, Theorem~\ref{theorem: new_cnot_count} presents high multiplicative constants due to the multi-controlled $SU(2)$ gates. As a result, initially, the approximate decomposition only outperforms the original exact decomposition for a higher number of qubits, although it outperforms the approximate decomposition from Lemma 7.8 of Ref~\cite{barenco_1995}. To further reduce the CNOT cost, an optimized multi-target multi-controlled $SU(2)$ decomposition is employed (Theorem~\ref{theorem: multi-target_su2}), making the approximated method more efficient than the exact decomposition for all number of control qubits up from $\nbase$ and also more efficient than the approximated decomposition proposed in Ref.~\cite{barenco_1995}. 

Multi-controlled quantum gates are a building block of many quantum algorithms. For example, in neural networks~\cite{de2023parametrized}, decomposition of unitary operators~\cite{iten2016quantum, malvetti2021quantum, shende2006synthesis}, quantum search~\cite{grover1997quantum}, quantum RAM memories~\cite{park2019circuit}, state initialization~\cite{zhang2021low, plesch2011quantum}, and others. The proposed method will reduce the elementary gates needed to run these algorithms without auxiliary qubits. In addition to improving the computational cost of applying a multi-controlled gate, our result is also significant in the context of NISQ devices. A smaller number of operations will reduce the noise from single-qubit and CNOT gates. 

One possible future research is to explore the potential impact of introducing additional hardware-level elementary operations to minimize the cost of decomposing multi-controlled gates. Such an investigation could pave the way for quantum hardware that accommodates software components efficiently. Perhaps this line of inquiry could yield an exact decomposition method for multi-controlled $U(2)$ gates with a linear cost without the need for auxiliary qubits or approximation.

\section*{Data availability}

The sites \url{https://github.com/qclib/qclib-papers} and \url{https://github.com/qclib/qclib}~\cite{qclib} contains all the data and the software generated during the current study. 

\section*{Conflict of interest}
The authors declare no conflicts of interest.

\bibliographystyle{IEEEtran}
\bibliography{refs}

\begin{thebibliography}{10}
\providecommand{\url}[1]{#1}
\csname url@samestyle\endcsname
\providecommand{\newblock}{\relax}
\providecommand{\bibinfo}[2]{#2}
\providecommand{\BIBentrySTDinterwordspacing}{\spaceskip=0pt\relax}
\providecommand{\BIBentryALTinterwordstretchfactor}{4}
\providecommand{\BIBentryALTinterwordspacing}{\spaceskip=\fontdimen2\font plus
\BIBentryALTinterwordstretchfactor\fontdimen3\font minus
  \fontdimen4\font\relax}
\providecommand{\BIBforeignlanguage}[2]{{%
\expandafter\ifx\csname l@#1\endcsname\relax
\typeout{** WARNING: IEEEtran.bst: No hyphenation pattern has been}%
\typeout{** loaded for the language `#1'. Using the pattern for}%
\typeout{** the default language instead.}%
\else
\language=\csname l@#1\endcsname
\fi
#2}}
\providecommand{\BIBdecl}{\relax}
\BIBdecl

\bibitem{preskill2018quantum}
J.~Preskill, ``Quantum computing in the {NISQ} era and beyond,''
  \emph{Quantum}, vol.~2, p.~79, 2018.

\bibitem{kim2023evidence}
Y.~Kim, A.~Eddins, S.~Anand, K.~X. Wei, E.~Van Den~Berg, S.~Rosenblatt,
  H.~Nayfeh, Y.~Wu, M.~Zaletel, K.~Temme \emph{et~al.}, ``Evidence for the
  utility of quantum computing before fault tolerance,'' \emph{Nature}, vol.
  618, no. 7965, pp. 500--505, 2023.

\bibitem{araujo2021divide}
I.~F. Araujo, D.~K. Park, F.~Petruccione, and A.~J. da~Silva, ``A
  divide-and-conquer algorithm for quantum state preparation,''
  \emph{Scientific reports}, vol.~11, no.~1, p. 6329, 2021.

\bibitem{zhang2021low}
X.-M. Zhang, M.-H. Yung, and X.~Yuan, ``Low-depth quantum state preparation,''
  \emph{Physical Review Research}, vol.~3, no.~4, p. 043200, 2021.

\bibitem{shende2006synthesis}
V.~Shende, S.~Bullock, and I.~Markov, ``Synthesis of quantum-logic circuits,''
  \emph{IEEE Transactions on Computer-Aided Design of Integrated Circuits and
  Systems}, vol.~25, no.~6, pp. 1000--1010, 2006.

\bibitem{iten2016quantum}
R.~Iten, R.~Colbeck, I.~Kukuljan, J.~Home, and M.~Christandl, ``Quantum
  circuits for isometries,'' \emph{Physical Review A}, vol.~93, no.~3, p.
  032318, 2016.

\bibitem{malvetti2021quantum}
E.~Malvetti, R.~Iten, and R.~Colbeck, ``Quantum circuits for sparse
  isometries,'' \emph{Quantum}, vol.~5, p. 412, 2021.

\bibitem{barenco_1995}
A.~Barenco, C.~H. Bennett, R.~Cleve, D.~P. DiVincenzo, N.~Margolus, P.~Shor,
  T.~Sleator, J.~A. Smolin, and H.~Weinfurter, ``Elementary gates for quantum
  computation,'' \emph{Physical Review A}, vol.~52, pp. 3457--3467, 1995.

\bibitem{nielsen2010quantum}
M.~A. Nielsen and I.~L. Chuang, \emph{Quantum computation and quantum
  information}.\hskip 1em plus 0.5em minus 0.4em\relax Cambridge university
  press, 2010.

\bibitem{brugiere2021reducing}
T.~G.~d. Brugière, M.~Baboulin, B.~Valiron, S.~Martiel, and C.~Allouche,
  ``Reducing the depth of linear reversible quantum circuits,'' \emph{IEEE
  Transactions on Quantum Engineering}, vol.~2, pp. 1--22, 2021.

\bibitem{bae2020quantum}
J.-H. Bae, P.~M. Alsing, D.~Ahn, and W.~A. Miller, ``Quantum circuit
  optimization using quantum {K}arnaugh map,'' \emph{Scientific reports},
  vol.~10, no.~1, p. 15651, 2020.

\bibitem{cuomo2023optimized}
D.~Cuomo, M.~Caleffi, K.~Krsulich, F.~Tramonto, G.~Agliardi, E.~Prati, and
  A.~S. Cacciapuoti, ``Optimized compiler for distributed quantum computing,''
  \emph{ACM Transactions on Quantum Computing}, vol.~4, no.~2, pp. 1--29, 2023.

\bibitem{leymann2020bitter}
F.~Leymann and J.~Barzen, ``The bitter truth about gate-based quantum
  algorithms in the {NISQ} era,'' \emph{Quantum Science and Technology},
  vol.~5, no.~4, p. 044007, 2020.

\bibitem{adenilton2022linear}
A.~J. da~Silva and D.~K. Park, ``Linear-depth quantum circuits for multiqubit
  controlled gates,'' \emph{Physical Review A}, vol. 106, p. 042602, 2022.

\bibitem{he2017decompositions}
Y.~He, M.-X. Luo, E.~Zhang, H.-K. Wang, and X.-F. Wang, ``Decompositions of
  n-qubit {T}offoli gates with linear circuit complexity,'' \emph{International
  Journal of Theoretical Physics}, vol.~56, pp. 2350--2361, 2017.

\bibitem{10.1007/978-3-031-08760-8_16}
S.~Balauca and A.~Arusoaie, ``Efficient constructions for simulating multi
  controlled quantum gates,'' in \emph{Computational Science -- ICCS 2022},
  Cham, 2022, pp. 179--194.

\bibitem{adenilton_su2_2023}
\BIBentryALTinterwordspacing
R.~Vale, T.~M.~D. Azevedo, I.~C.~S. Araújo, I.~F. Araujo, and A.~J. da~Silva,
  ``Decomposition of multi-controlled special unitary single-qubit gates,''
  2023. [Online]. Available: \url{https://arxiv.org/abs/2302.06377v1}
\BIBentrySTDinterwordspacing

\bibitem{saeedi2013linear}
M.~Saeedi and M.~Pedram, ``Linear-depth quantum circuits for n-qubit {T}offoli
  gates with no ancilla,'' \emph{Physical Review A}, vol.~87, no.~6, p. 062318,
  2013.

\bibitem{maslov2016advantages}
D.~Maslov, ``Advantages of using relative-phase toffoli gates with an
  application to multiple control toffoli optimization,'' \emph{Physical Review
  A}, vol.~93, no.~2, p. 022311, 2016.

\bibitem{ibm_mumbai}
I.~Q. Experience, ``{IBM} {M}umbai,''
  \url{https://quantum-computing.ibm.com/services/resources?tab=systems&system=ibmq_mumbai},
  accessed March 14, 2023.

\bibitem{welch2014efficient}
J.~Welch, D.~Greenbaum, S.~Mostame, and A.~Aspuru-Guzik, ``Efficient quantum
  circuits for diagonal unitaries without ancillas,'' \emph{New Journal of
  Physics}, vol.~16, no.~3, p. 033040, 2014.

\bibitem{wille2013improving}
R.~Wille, M.~Soeken, C.~Otterstedt, and R.~Drechsler, ``Improving the mapping
  of reversible circuits to quantum circuits using multiple target lines,'' in
  \emph{2013 18th Asia and South Pacific Design Automation Conference
  (ASP-DAC)}.\hskip 1em plus 0.5em minus 0.4em\relax IEEE, 2013, pp. 145--150.

\bibitem{de2023parametrized}
J.~H. de~Carvalho and F.~M. de~Paula~Neto, ``Parametrized constant-depth
  quantum neuron,'' \emph{IEEE Transactions on Neural Networks and Learning
  Systems}, 2023.

\bibitem{grover1997quantum}
L.~K. Grover, ``Quantum mechanics helps in searching for a needle in a
  haystack,'' \emph{Physical review letters}, vol.~79, no.~2, p. 325, 1997.

\bibitem{park2019circuit}
D.~K. Park, F.~Petruccione, and J.-K.~K. Rhee, ``Circuit-based quantum random
  access memory for classical data,'' \emph{Scientific reports}, vol.~9, no.~1,
  p. 3949, 2019.

\bibitem{plesch2011quantum}
M.~Plesch and {\v{C}}.~Brukner, ``Quantum-state preparation with universal gate
  decompositions,'' \emph{Physical Review A}, vol.~83, no.~3, p. 032302, 2011.

\bibitem{qclib}
I.~F. Araujo, I.~C.~S. Araújo, L.~D. da~Silva, C.~Blank, and A.~J. da~Silva,
  ``{Quantum computing library},'' \url{https://github.com/qclib/qclib}, Feb.
  2023, version 0.0.21.

\end{thebibliography}
\end{document}